\documentclass[letterpaper, 10 pt, conference]{ieeeconf}  
\usepackage{listings,color,caption}

\IEEEoverridecommandlockouts
\usepackage{amsthm, amssymb, mathtools, graphicx}
\usepackage{amsmath}
\usepackage{bm}
\usepackage{cite}
\usepackage{arydshln}
\usepackage{verbatim}

\usepackage{subcaption}

\usepackage{etoolbox}
\newbool{arxivversion}
\setbool{arxivversion}{true}

\usepackage{algpseudocode}
\usepackage{algorithm}

\newtheorem{assumption}{\bf Assumption}
\newtheorem{definition}{\bf Definition}
\newtheorem{theorem}{\bf Theorem}
\newtheorem{proposition}{\bf Proposition}
\newtheorem{lemma}{\bf Lemma}

\def\JV{\textcolor{black}}

\begin{document}

\title{Beyond Asymptotics: Targeted exploration with finite-sample guarantees}

\author{Janani Venkatasubramanian$^1$, Johannes K\"ohler$^2$ and Frank Allg\"ower$^1$
\thanks{$^1$ Janani Venkatasubramanian and Frank Allg\"ower are with the Institute for Systems Theory and Automatic Control, University of Stuttgart, 70550 Stuttgart, Germany. (email:\{janani.venkatasubramanian, frank.allgower\}@ist.uni-stuttgart.de).}
\thanks{$^2$ Johannes K\"ohler is with the Institute for Dynamic Systems and Control, ETH Z\"urich, Z\"urich CH-80092, Switzerland. (email:jkoehle@ethz.ch)}
\thanks{F. Allg\"ower is thankful that his work was funded by Deutsche Forschungsgemeinschaft (DFG, German Research Foundation) under Germany’s Excellence Strategy - EXC 2075 - 390740016 and under grant 468094890. F. Allg\"ower acknowledges the support by the Stuttgart Center for Simulation Science (SimTech). J. Venkatasubramanian thanks the International Max Planck Research School for Intelligent Systems (IMPRS-IS) for supporting her.}}

\maketitle

\begin{abstract}
In this paper, we introduce a targeted exploration strategy for the non-asymptotic, finite-time case. The proposed strategy is applicable to uncertain linear time-invariant systems subject to sub-Gaussian disturbances. As the main result, the proposed approach provides a priori guarantees, ensuring that the optimized exploration inputs achieve a desired accuracy of the model parameters. The technical derivation of the strategy (i) leverages existing non-asymptotic identification bounds with self-normalized martingales, (ii) utilizes spectral lines to predict the effect of sinusoidal excitation, and (iii) effectively accounts for spectral transient error and parametric uncertainty. A numerical example illustrates how the finite exploration time influences the required exploration energy. 

\end{abstract}


\section{Introduction}
System identification bridges control theory and machine learning by providing methods to model dynamical systems from data, enabling better prediction and control \cite{ljung1999sysid}. Within this field, optimal experiment design and targeted exploration~\cite{bombois2011optimal} develop methods to optimally excite dynamical systems to reduce model uncertainty, thereby (i) attaining a model with desired accuracy \cite{jansson2005input, bombois2021robust}, or (ii) ensuring the feasibility of a subsequent robust control design~\cite{barenthin2008identification, umenberger2019robust, ferizbegovic2019learning, venkatasubramanian2024robust, venkatasubramanian2023sequential}. Theoretical results on optimal experiment design rely primarily on asymptotic bounds for identification error, which are valid in the limit as the amount of data approaches infinity~\cite{ljung1999sysid}. However, in practice we only have finite data from experiments and the reliability of the model-based controller depends significantly on the quality of the data used for system identification. Recently, numerous results have been developed for non-asymptotic (finite-sample) analysis of system identification \cite{simchowitz2018learning, sarkar2019near, tsiamis2019finite} and subsequent control \cite{abbasi2011improved, dean2017sample, mania2019certainty}. An overview of these methods is provided in \cite{tsiamis2023statistical}.
In general, these non-asymptotic results depend on excitation from applying random inputs to ensure learning, and do not optimize the inputs for excitation. Recent works \cite{wagenmaker2020active, sarker2020parameter} consider targeted exploration with periodic/sinusoidal inputs, adopting a frequency-domain analysis for parameter estimation.
However, neither approach explicitly accounts for transient effects~\cite{ljung1999sysid} during input design or analysis. Hence, their guarantees are contingent on the dissipation of transient error and the attainment of steady-state response. This leads to prolonged experiment durations, which may be impractical for many real-world applications.
 
In this paper, we design a targeted exploration strategy that ensures a desired error bound on the estimated parameters with high probability by utilizing a non-asymptotic data-dependent uncertainty bound \cite{abbasi2011improved, sarkar2019near}. In particular, we derive \textit{a priori} guarantees on the error bound that can be ensured before exploration. To shape and reduce uncertainty in a targeted manner, we consider multi-sine exploration inputs in selected frequencies and optimized amplitudes. As our main contribution, we derive sufficient conditions directly in terms of the exploration inputs that ensure a desired parameter error bound with high probability through finite-time exploration. In particular, we explicitly account for transient error due to the input and the effect of the disturbances in the design of targeted exploration. We use these sufficient conditions to formulate linear matrix inequalities (LMIs) for exploration, ensuring the desired parameter error bound. This leads to a targeted exploration design with minimal input energy based on a semidefinite program (SDP).

\textit{Outline: }We first provide the setting and define the exploration objective in Section~\ref{sec:problemstatement}. We then summarize data-dependent non-asymptotic identification bounds from \cite{abbasi2011improved, sarkar2019near} in Section \ref{sec:nonasymp_ub}. In Section \ref{sec:suffcond_data}, we derive sufficient conditions on the exploration data. In Section \ref{sec:suffcond_spectral} we derive sufficient conditions on the exploration inputs by leveraging the theory of spectral lines \cite{sarker2020parameter} and explicitly accounting for transient error due to the input in the exploration design. Finally, in Section \ref{sec:numerical}, we provide a numerical example to highlight the effect of the finite time of the experiment on the required input energy.

\textit{Notation: }The transpose (conjugate transpose) of a matrix $A \in \mathbb{R}^{n \times m} (\mathbb{C}^{n \times m})$ is denoted by $A^\top$ ($A^\mathsf{H}$).
The positive semi-definiteness (definiteness) of a matrix $A \in \mathbb{C}^{n \times n}$ is denoted by $A = A^\mathsf{H} \succeq (\succ) 0$. The Kronecker product operator is denoted by $\otimes$. The operator $\mathrm{vec}(A)$ stacks the columns of $A$ to form a vector. The operator $\textnormal{diag}(A_1,\dots,A_n)$ creates a block diagonal matrix by aligning the matrices $A_1,\dots,A_n$ along the diagonal. The Euclidean norm for a vector $x\in\mathbb{R}^n$ is denoted by $\|x\|=\sqrt{x^\top x}$. Given a matrix $P \in \mathbb{R}^{n \times n},\, P \succ 0$, the weighted Euclidean norm is given by $\|x\|_P=\sqrt{x^\top P x}$.  The largest singular value of a matrix $A\in\mathbb{C}^{m\times n}$ is denoted by $\|A\|$. \JV{Given a matrix $M \in \mathbb{C}^{m \times m},\, M \succeq 0$, the weighted Frobenius norm of a matrix $A\in\mathbb{C}^{m\times n}$ is denoted by $\|A\|_{F,M}=\sqrt{\mathrm{tr}\left(A^\mathsf{H} M A\right)}$.} $\mathbb{E}[X]$ denotes the expected value of a random variable $X$. $\mathbb{P}[E]$ denotes the probability of an event $E$ occurring.




\section{Problem Statement}\label{sec:problemstatement}
Consider a discrete-time linear time-invariant dynamical system of the form \begin{equation}\label{eq:sys}
x_{k+1}=A_{\mathrm{tr}}x_k+B_{\mathrm{tr}}u_k+w_k,
\end{equation}
where time $k \in \mathbb{N}$, $x_k \in \mathbb{R}^{n_\mathrm{x}}$ is the state, $u_k \in \mathbb{R}^{n_\mathrm{u}}$ is the control input, and $w_k \in \mathbb{R}^{n_\mathrm{x}}$ is the disturbance. It is assumed the state $x_k$ is directly measurable. In order to simplify the exposition, we assume that the initial state $x_0$ is zero. The true system parameters $A_\mathrm{tr}$, $B_\mathrm{tr}$ are initially uncertain, and it is necessary to collect informative data through the process of exploration to enhance the accuracy of the parameters. 

Henceforth, we denote $\phi_k=[x_k^\top\,, u_k^\top]^\top \in \mathbb{R}^{n_\phi}$ where $n_\phi=n_\mathrm{x}+n_\mathrm{u}$. The system \eqref{eq:sys} can be re-written in terms of the unknown parameter $\theta_\mathrm{tr}=\textnormal{vec}([A_\mathrm{tr},B_\mathrm{tr}]) \in \mathbb{R}^{n_\theta}$ as
 \begin{align}\label{eq:systheta}
 x_{k+1}=(\phi_k^\top \otimes I_{n_\mathrm{x}})\theta_\mathrm{tr} + w_k
 \end{align}
 where $n_\theta = n_\mathrm{x}n_\phi$. We consider a filtration of $\sigma$-algebras $\{\mathcal{F}_k \}_{k \geq 0}$ such that $\phi_k$ is $\mathcal{F}_{k-1}$-measurable and $w_k$ is $\mathcal{F}_k$-measurable.
\begin{definition}(Sub-Gaussian random vector \cite[Definition 2]{ao2025stochastic})
    A random vector $w_k \in \mathbb{R}^{n_\mathrm{x}}$ with mean $\mathbb{E}[w_k]=\mu$ is called sub-Gaussian with proxy variance $\Sigma \succeq 0$, i.e., $w_k \sim \mathsf{subG}(\mu,\Sigma)$, if $\forall \lambda \in \mathbb{R}^{n_\mathrm{x}}$,
\begin{align*}
    \mathbb{E}\left[ \exp (\lambda^\top (w_k-\mu))|\mathcal{F}_{k-1}\right] \leq \exp \left( \frac{\|\lambda\|_{\Sigma}^2}{2} \right).
\end{align*}
\end{definition}
\begin{assumption}\label{a1}
    The disturbance $w_k$ is conditionally sub-Gaussian with zero mean and known proxy variance $\sigma_\mathrm{w}^2I_{n_\mathrm{x}}$ \cite{abbasi2011improved}, i.e., $w_k \sim \mathsf{subG}(0,\sigma_\mathrm{w}^2I_{n_\mathrm{x}})$. 
\end{assumption}

The sub-Gaussian distributions encompass Gaussian, uniform, and other distributions with light tails. We assume that some prior knowledge on the dynamics is available.
\begin{assumption}\label{a3}
    The system matrix $A_\mathrm{tr}$ is Schur stable and the pair $(A_\mathrm{tr},B_\mathrm{tr})$ is controllable.
\end{assumption}
\begin{assumption} \label{a2} The parameters $\theta_\mathrm{tr}=\mathrm{vec}([A_\mathrm{tr},B_\mathrm{tr}]) \in \mathbb{R}^{n_\theta}$ lie in a known set $\mathbf{\Theta}_0$, i.e., $\theta_\mathrm{tr} \in \mathbf{\Theta}_0$, where 
\begin{align}\label{eq:Theta0}
\mathbf{\Theta}_0:=\left\{\theta:(\hat{\theta}_0-\theta)^\top (D_0 \otimes I_{n_\mathrm{x}}) (\hat{\theta}_0-\theta) \leq 1 \right\},
\end{align}
with an estimate $\hat{\theta}_0$ and for some $D_0 \succ 0$.
\end{assumption}

\textit{Exploration goal:} Since the true system parameters $\theta_\mathrm{tr}~=~\text{vec}([A_\mathrm{tr}$, $B_\mathrm{tr}])$ are not precisely known, there is a necessity to gather informative data through the process of exploration. Our primary goal is to design exploration inputs that excite the system, for a fixed user chosen time $T \in \mathbb{N}$, in a manner as to obtain an estimate $\hat{\theta}_T=\text{vec}([\hat{A}_T,\hat{B}_T])$ that satisfies 
\begin{align}\label{eq:exp_goal2}
    (\theta_\mathrm{tr}-\hat{\theta}_T)^\top \left( D_\mathrm{des} \otimes I_\mathrm{n_x} \right)(\theta_\mathrm{tr}-\hat{\theta}_T)\leq 1.
\end{align}  
Here, $D_\mathrm{des} \succ 0$ is a user-defined matrix characterizing how closely $\hat{\theta}_T$ approximates the true parameters $\theta_\mathrm{tr}$. The exploration inputs are computed such that it excites the system sufficiently with minimal input energy, based on the initial parameter estimates (cf. Assumption \ref{a2}). Denote $U=[u_0^\top,\cdots,u_{T-1}^\top]^\top \in \mathbb{R}^{T n_\mathrm{u}}$. Bounding the input energy by a constant $T \gamma_\mathrm{e}^2\geq 0$ can be equivalently written as
\begin{align}\label{eq:gammau}
    \sum_{k=0}^{T-1}\|u_k\|^2 = \|U\|^2\leq T \gamma_\mathrm{e}^2.
\end{align}



\section{Preliminaries: Non-asymptotic uncertainty bound}\label{sec:nonasymp_ub}
In this section, we focus on quantifying the uncertainty associated with the unknown parameters $\theta_\mathrm{tr}=\textnormal{vec}([A_{\mathrm{tr}},B_{\mathrm{tr}}])$. For this purpose, we first outline regularized least squares estimation, and then summarize existing results by \cite{abbasi2011improved, sarkar2019near} that provide a data-dependent uncertainty bound on the uncertain parameters. 

Given observed data $\mathcal{D}_{T+1}=\{x_k,u_k\}_{k=0}^{T}$ of length $T+1$, we denote
\begin{align}\label{eq:phi_compact}
\nonumber
    \Phi & = [\phi_0,\dots,\phi_{T-1}] \in \mathbb{R}^{n_\phi\times T},\\
\nonumber
    X & = [x_1^\top,\dots, x_T^\top]^\top \in \mathbb{R}^{Tn_\mathrm{x} \times 1},\\
    W & = [w_0^\top,\dots, w_{T-1}^\top]^\top \in \mathbb{R}^{Tn_\mathrm{x}  \times 1}.
\end{align}
Given \eqref{eq:phi_compact}, the regularized least squares estimate is given by
\begin{align}\label{eq:reg_mean_est}
    \hat{\theta}_{T}=\underbrace{( (D_T \otimes I_{n_\mathrm{x}})+\lambda I_{n_\phi n_\mathrm{x}})^{-1}}_{=:\bar{D}_T^{-1}} (\Phi\otimes I_{n_\mathrm{x}}) X,
\end{align} 
with the regularization constant $\lambda > 0$ and the excitation
\begin{align}\label{eq:covar_est}
    D_T= \Phi \Phi^\top.
\end{align}
Using \eqref{eq:reg_mean_est}, the estimate $\hat{\theta}_{T}$ can be expressed as
\begin{align}\label{eq:mean_est_rls}
\nonumber
\hat{\theta}_{T} \overset{\eqref{eq:systheta}}{=}& \bar{D}_T^{-1} (\Phi\otimes I_{n_\mathrm{x}})((\Phi^\top \otimes I_{n_\mathrm{x}})\theta_\mathrm{tr} + W)\\
=&\bar{D}_T^{-1} ((D_T \otimes I_{n_\mathrm{x}}) \theta_\mathrm{tr} + (\Phi \otimes I_{n_\mathrm{x}})W).
\end{align}
Thus, the regularized least squares error satisfies
\begin{align}\label{eq:errtheta}
    \hat{\theta}_{T}-\theta_\mathrm{tr}=& -\lambda \bar{D}_T^{-1} \theta_\mathrm{tr}
     + \bar{D}_T^{-1} \underbrace{(\Phi\otimes I_{n_\mathrm{x}}) W}_{=:S_T}.
\end{align}
 \JV{Let $\tilde{S}_T=\sum_{k=0}^{T-1} w_k \phi_k^\top$ so that $S_T=\mathrm{vec}(\tilde{S}_T)$.}
 
\textit{Self-normalized Martingales:} \JV{The sequence $\{\tilde{S}_T\}_{ T \geq 0}$ is a martingale with respect to $\{\mathcal{F}_T \}_{T=0}^\infty$  \cite{williams1991probability}.} This sequence is crucial for the construction of confidence ellipsoids for $\theta_\mathrm{tr}$. In \cite[Theorem 1]{abbasi2011improved}, a `self-normalized bound' for the martingale $\{\tilde{S}_T \}_{T \geq 0}$ is derived by assuming scalar disturbances $w_k$, i.e., $n_\mathrm{x}=1$


The following lemma generalizes the bound in \cite[Theorem 1]{abbasi2011improved} to vector-valued disturbances, i.e., $n_\mathrm{x} > 1$, by utilizing covering techniques as in \cite{tsiamis2023statistical}.
\JV{\begin{lemma} \label{lem:self-norm-bound-vector} (Self-normalized bound for vector processes \cite[Eq. (S6)]{tsiamis2023statistical}) Let Assumption \ref{a1} hold. For any $\delta \in (0,1)$, with probability at least $1-\delta$, for all $T \in \mathbb{N}$,
    \begin{align}\label{eq:nonreg-vector}
\|S_T\|_{\bar{D}_T^{-1}}^2 = \|\tilde{S}_T\|_{F,\bar{D}_T^{-1}}^2\leq \underbrace{2 n_\mathrm{x} \sigma_\mathrm{w}^2 \log \left(
\frac{n_\mathrm{x} \det (\bar{D}_T)^{\frac{1}{2}}}{\delta \det(\lambda I_{n_\mathrm{x} n_\phi})^{\frac{1}{2}}} \right)}_{:=R(\bar{D}_T)}.
\end{align}
\end{lemma}}
In order to derive the confidence ellipsoid, we first derive a bound on $\theta_\mathrm{tr}$ of the form $\|\theta_\mathrm{tr}\| \leq \bar{\theta}$. By applying the Schur complement twice to the condition in \eqref{eq:Theta0}, we get $(\hat{\theta}_0-\theta_\mathrm{tr})(\hat{\theta}_0-\theta_\mathrm{tr})^\top \preceq (D_0 \otimes I_{n_\mathrm{x}})^{-1}$, from which we have
\begin{align}\label{eq:theta0diffnorm}
    \|\hat{\theta}_0-\theta_\mathrm{tr}\| \leq \|D_0^{-\frac{1}{2}}\|.
\end{align}
Using the triangle inequality, we have
\begin{align}\label{eq:bartheta}
    \|\theta_\mathrm{tr}\| & \leq \|\hat{\theta}_0\| +  \|\hat{\theta}_0-\theta_\mathrm{tr}\| \overset{\eqref{eq:theta0diffnorm}}{\leq}  \|\hat{\theta}_0\| + \|D_0^{-\frac{1}{2}}\|=: \bar{\theta}.
\end{align}
The following theorem utilizes the result in Lemma~\ref{lem:self-norm-bound-vector} to derive a data-dependent uncertainty bound in the form of a confidence ellipsoid centered at the estimate $\hat{\theta}_{T}$ for the uncertain parameters $\theta_\mathrm{tr}$.

\begin{theorem}\label{thm:RLS_ellipsoid}\cite[Theorem 2]{abbasi2011improved}
Let Assumptions \ref{a1} and \ref{a2} hold. Then, for any $\delta > 0$, $\mathbb{P}[\theta_\mathrm{tr} \in \Theta_T, \forall T] \geq 1-\delta$ where
\begin{align}\label{eq:nonreg_ellipsoid}
    \Theta_T =\left\{ \theta: \, \|\hat{\theta}_{T} - \theta_\mathrm{tr} \|_{\bar{D}_T} \leq R(\bar{D}_T)^\frac{1}{2} + \lambda^\frac{1}{2} \bar{\theta} \right\}.
\end{align}
\end{theorem}
We utilize this data-dependent uncertainty bound to design a targeted exploration strategy that achieves the exploration goal provided in \eqref{eq:exp_goal2}.

\section{Data-dependent sufficient conditions for targeted exploration}\label{sec:suffcond_data}
Given the uncertainty bound in Theorem \ref{thm:RLS_ellipsoid}, the following theorem presents conditions on the exploration data that imply the exploration goal \eqref{eq:exp_goal2}.

\begin{theorem}\label{thm:RLS_suffcond} Let Assumptions \ref{a1} and \ref{a2} hold. Suppose $\Phi$ satisfies
\begin{align}\label{eq:suffcond_rls}
    \Phi \Phi^\top + \lambda I_{n_\phi}-\left(R(\bar{D}_T)^\frac{1}{2} + \lambda^\frac{1}{2} \bar{\theta} \right)^2  D_\mathrm{des}  \succeq 0
\end{align}
with $\bar{D}_T$ according to \eqref{eq:reg_mean_est}. Then, the estimate $\hat{\theta}_{T}$ computed as in \eqref{eq:mean_est_rls} satisfies the exploration goal \eqref{eq:exp_goal2} with probability at least $1-\delta$.
\end{theorem}
\ifbool{arxivversion}{The proof of Theorem \ref{thm:RLS_suffcond} may be found in Appendix \ref{app:thm_RLS_suffcond}.}{The proof of Theorem \ref{thm:RLS_suffcond} may be found in \cite[Appendix I]{venkatasubramanian2025beyond}.} In contrast to standard asymptotic bounds, e.g., \cite{umenberger2019robust, venkatasubramanian2023sequential}, Theorem \ref{thm:RLS_suffcond} depends additionally on the term $(R(\bar{D}_T)^\frac{1}{2} + \lambda^\frac{1}{2} \bar{\theta})^2$ which is  nonconvex in $\bar{D}_T$. Hence, in the following lemma we derive a sufficient condition that is convex in $\bar{D}_T$.
\JV{\begin{lemma}\label{lem:boundRDT}
For any $\bar{D}_T \in \mathbb{R}^{n_\theta}$, we have
\begin{align}\label{eq:bound_RDT_lambda}
    (R(\bar{D}_T)^\frac{1}{2} + \lambda^\frac{1}{2} \bar{\theta})^2  \leq  (2C_1 + 2 n_\mathrm{x} \sigma_\mathrm{w}^2 \log(\det(\bar{D}_T)) + 2 \lambda \bar{\theta}^2 ) 
\end{align}
where
\begin{align}\label{eq:C1}
    C_1:= n_\mathrm{x} \sigma_\mathrm{w}^2 \left( \log \left( \frac{n_\mathrm{x}^2}{\delta^{2}}\right) -n_\theta\log (\lambda) \right).
\end{align} 
\end{lemma}}
\ifbool{arxivversion}{The proof of Lemma \ref{lem:boundRDT} may be found in Appendix \ref{app:lem_boundRDT}.}{The proof of Lemma \ref{lem:boundRDT} may be found in \cite[Appendix II]{venkatasubramanian2025beyond}.}
\begin{proposition}\label{prop:RLS_suffcond_RDTbound} Let Assumptions \ref{a1} and \ref{a2} hold. Suppose $\Phi$ satisfies
\begin{align}\label{eq:suffcond_prop}
\nonumber
    \Phi \Phi^\top + \lambda I_{n_\phi}&\\
    -\left( 2C_1 + \JV{2 n_\mathrm{x}} \sigma_\mathrm{w}^2 \log(\det(\bar{D}_T)) + 2 \lambda \bar{\theta}^2 \right) D_\mathrm{des} &\succeq 0.
\end{align}
Then, the estimate $\hat{\theta}_{T}$ computed as in \eqref{eq:mean_est_rls} satisfied the exploration goal \eqref{eq:exp_goal2} with probability at least $1-\delta$.
\end{proposition}

\ifbool{arxivversion}{The proof of Proposition \ref{prop:RLS_suffcond_RDTbound} may be found in Appendix \ref{app:prop_RLS_sc}.}{The proof of Proposition \ref{prop:RLS_suffcond_RDTbound} may be found in \cite[Appendix III]{venkatasubramanian2025beyond}.} Note that Inequality \eqref{eq:suffcond_prop} depends quadratically on $\Phi$, specifically on the finite excitation term $\Phi \Phi^\top$ \cite{sarker2020parameter}. This further depends on the amplitudes of the inputs $U_\mathrm{e}$ and the disturbances. Additionally, the linear mapping from the input sequence to the state sequence is uncertain due to the uncertainty in the true dynamics. These issues are addressed in the next section.

\section{Sufficient conditions for targeted exploration based on spectral lines}\label{sec:suffcond_spectral}
In this section, we determine sufficient conditions for exploration in terms of the spectral content of $\Phi$ in order to pose the exploration problem directly in terms of the inputs. To this end, we first define the amplitude of a spectral line.
\begin{definition}
(Amplitude of a spectral line \cite{sarker2020parameter}) Given a sequence $\{\phi_k\}_{k=0}^{T-1}$, the amplitude of the spectral line of the sequence $\bar{\phi}(\omega_0)$ at a frequency $\omega_0 \in \Omega_T:=\{0,1/T,\dots,(T-1)/T\}$ is given by
\begin{align}\label{eq:spectral_amplitude}
 \bar{\phi}(\omega_0):=\frac{1}{T}\sum_{k=0}^{T-1}\phi_k e^{-j2\pi \omega_0 k}.
\end{align}
\end{definition}
In what follows, we first establish the exploration strategy in Section \ref{sec:expstrategy} and then provide a bound on the excitation $\Phi \Phi^\top$ in terms of spectral information in Section \ref{sec:excitation}. We then explicitly account for transient error due to the input and the effect of disturbances in Section \ref{sec:te_dist}. In order to derive conditions linear in the decision variables, we provide a convex relaxation procedure in Section \ref{sec:convexrel} and an upper bound on excitation in Section \ref{sec:ubDT}. Finally, we account for parametric uncertainty and obtain an SDP for exploration, which provides us with exploration inputs that ensure the exploration goal in Section \ref{sec:expLMI}.

\subsection{Exploration strategy}\label{sec:expstrategy}
The exploration input sequence takes the form
\begin{align}\label{eq:exploration_controller}
u_k=\sum_{i=1}^{L} \bar{u}(\omega_i) \cos(2 \pi \omega_i k), \quad\, k=0,\dots,T-1
\end{align}
where $\bar{u}(\omega_i) \in \mathbb{R}^{n_\mathrm{u}}$ are the amplitudes of the sinusoidal inputs at $L \in \mathbb{N}$ distinct selected frequencies $\omega_i \in \Omega_T$.
The amplitude of the spectral line of the sequence $\{u_k\}_{k=0}^{T-1}$ at frequency $\omega_i$ is $\bar{u}(\omega_i)$. Denote $U_\mathrm{e}=\mathrm{diag}(\bar{u}(\omega_1),\dots,\bar{u}(\omega_L)) \in \mathbb{R}^{Ln_\mathrm{u} \times L}$.
By the Parseval-Plancheral identity, bounding the average input energy by a constant $\gamma^2_\mathrm{e}$ \eqref{eq:gammau} can be equivalently written as
\begin{align}\label{eq:upars}
\frac{1}{T}\|U\|^2 =\sum_{i=1}^L \|\bar{u}(\omega_i)\|^2= \mathbf{1}_{L}^\top U_\mathrm{e}^\top U_\mathrm{e} \mathbf{1}_{L} \leq \gamma_\mathrm{e}^2
\end{align}
where $\mathbf{1}_{L} \in \mathbb{R}^{L\times 1}$ is a vector of ones, and the bound $\gamma_\mathrm{e} \geq 0$ is desired to be small. Using the Schur complement, this criterion is equivalent to the following LMI:
\begin{equation}\label{eq:min_energy_cost}
S_{\textnormal{energy-bound-1}}(\gamma_\mathrm{e},U_\mathrm{e})\coloneqq \begin{bmatrix} 
\gamma_\mathrm{e} & \mathbf{1}_{L}^\top U_\mathrm{e}^\top \\ U_\mathrm{e} \mathbf{1}_{L} & \gamma_\mathrm{e} I
\end{bmatrix} \succeq 0.
\end{equation}

\subsection{Bound on excitation based on spectral lines}\label{sec:excitation}
The uncertainty bound on the parameters of the system evolving under the exploration input sequence \eqref{eq:exploration_controller} can be computed from the spectral content of the observed state $x_k$ and the applied input $u_k$. Given $u_k$ as in \eqref{eq:exploration_controller}, $x_k$ and $u_k$ have $L$ spectral lines from $0$ to $T-1$ at distinct frequencies $\omega_i \in \Omega_T$, $i=1,\dots,L$ with amplitudes $\bar{x}(\omega_i)$ and $\bar{u}(\omega_i)$. Recall $\phi_k=[x_k^\top\,u_k^\top]^\top$. From \eqref{eq:sys}, we have
\begin{align}\label{eq:phibar}
    \bar{x}(\omega_i) & = \underbrace{\left[ (e^{j\omega_i}I-A_\mathrm{tr})^{-1}B_\mathrm{tr}\right] \bar{u}(\omega_i)}_{\bar{x}_\mathrm{u,d}(\omega_i)} + \bar{x}_\mathrm{u,t}(\omega_i) + \bar{x}_\mathrm{w}(\omega_i),
\end{align}
where $\bar{x}_{\mathrm{u,d}}(\omega_i)$ denotes the asymptotic effect of the input, $\bar{x}_{\mathrm{u,t}}(\omega_i)$ denotes the transient error and $\bar{x}_{\mathrm{w}}(\omega_i)$ denotes the effect of the disturbance. Furthermore, we have
\begin{align}
\nonumber
    \bar{x}_\mathrm{u}(\omega_i)
    =& \bar{x}_\mathrm{u,d}(\omega_i)+\bar{x}_\mathrm{u,t}(\omega_i)\\
\nonumber
    =&\frac{1}{T}\sum_{k=0}^{T-1} \left( \sum_{l=0}^k A_\mathrm{tr}^{k-l}B_\mathrm{tr} u_l \right) e^{-j2 \pi \omega_i k}\\
\nonumber
    =& \underbrace{\frac{1}{T}\left(\begin{bmatrix}
    e^{-j2 \pi \omega_i 0} & \cdots & e^{-j2 \pi \omega_i (T-1)}
\end{bmatrix} \otimes I_{n_\mathrm{x}} \right)}_{=:F_{x,\omega_i}}\\
\label{eq:barphiu_redef}
    &\times \underbrace{\begin{bmatrix}
        B_\mathrm{tr} & 0 &  \cdots & \vdots\\
        A_\mathrm{tr}B_\mathrm{tr} & B_\mathrm{tr} & \cdots & \vdots \\
        \vdots & \cdots &  B_\mathrm{tr}& 0 \\
        A_\mathrm{tr}^{T-1}B_\mathrm{tr} & \cdots &  A_\mathrm{tr}B_\mathrm{tr} & B_\mathrm{tr}
    \end{bmatrix}}_{A_\mathrm{u}}
    U,\\
    \nonumber
    \bar{x}_\mathrm{w}(\omega_i)=&\frac{1}{T}\sum_{k=0}^{T-1} \left( \sum_{l=0}^k A_\mathrm{tr}^{k-l} w_l \right) e^{-j2 \pi \omega_i k}\\
\label{eq:barphiw_redef}
    =&F_{x,\omega_i} \underbrace{\begin{bmatrix}
        I & 0 & \cdots & \vdots\\
        A_\mathrm{tr} & I & \cdots & \vdots \\
        \vdots & \cdots &  I& 0 \\
        A_\mathrm{tr}^{T-1} & \cdots & A_\mathrm{tr} & I
    \end{bmatrix}}_{A_\mathrm{w}} W.
\end{align}
Further, $\bar{u}(\omega_i)$ can be written as
\begin{align}
    \bar{u}(\omega_i)=\underbrace{\frac{1}{T}\left(\begin{bmatrix}
    e^{-j2 \pi \omega_i 0} & \cdots & e^{-j2 \pi \omega_i (T-1)}
\end{bmatrix} \otimes I_{n_\mathrm{u}} \right)}_{=:F_{u,\omega_i}}U.
\end{align}

Note that for all $\omega_i \in \Omega_T$, we have
\begin{align}\label{eq:supFwi}
     \|F_{x,\omega_i}\|^2=\|F_{u,\omega_i}\|^2 = \frac{1}{T}.
\end{align}
Using \eqref{eq:phibar}, we have
\begin{align}
    \bar{\phi}(\omega_i)&=\underbrace{\left[\bar{x}_\mathrm{u,d}(\omega_i) \atop \bar{u}(\omega_i)\right]}_{\bar{\phi}_\mathrm{u,d}(\omega_i)}+\underbrace{\left[\bar{x}_\mathrm{u,t}(\omega_i) \atop 0\right]}_{\bar{\phi}_\mathrm{u,t}(\omega_i)}+\underbrace{\left[\bar{x}_\mathrm{w}(\omega_i) \atop 0\right]}_{\bar{\phi}_\mathrm{w}(\omega_i)},
\end{align}
where
\begin{align}
\bar{\phi}_\mathrm{u,d}(\omega_i)=\underbrace{\begin{bmatrix}
(e^{j \omega_i}I - A_\mathrm{tr})^{-1}B_\mathrm{tr} \\ I_{n_\mathrm{u}}
\end{bmatrix}}_{=:V_i} \bar{u}(\omega_i).
\end{align}
We compactly define
\begin{align}
    \bar{\Phi}_\ast& =[\bar{\phi}_\ast (\omega_1), \dots ,\bar{\phi}_\ast (\omega_L)] \in \mathbb{C}^{n_\phi \times L},
\end{align}
where $\ast \in \{\mathrm{w};\mathrm{u};\mathrm{u,d};\mathrm{e,t} \}$, which satisfies
\begin{align}\label{eq:phibarcompact}
\nonumber
    \bar{\Phi} & = \underbrace{V_\mathrm{tr}U_\mathrm{e}  + \bar{\Phi}_\mathrm{u,t}}_{=:\bar{\Phi}_\mathrm{u}}  + \bar{\Phi}_\mathrm{w},\\
    V_{\mathrm{tr}}&:=[V_{1},\cdots,V_{L}] \in \mathbb{C}^{n_\mathrm{\phi}\times n_\mathrm{u}L}.
\end{align}

The following lemma provides a lower bound on $\Phi \Phi^\top$ in terms of the spectral lines of $\phi_k$.
\begin{lemma}\label{lem:phiphitight} For any $\epsilon \in (0,1)$, $\phi_k$ satisfies
\begin{align}\label{eq:phiphitight}
  \frac{1}{T}\Phi \Phi^\top\succeq & (1-\epsilon)V_\mathrm{tr}U_\mathrm{e}U_\mathrm{e}^\top V_\mathrm{tr}^\mathsf{H}-\tfrac{2(1-\epsilon)}{\epsilon}(\bar{\Phi}_\mathrm{u,t}\bar{\Phi}_\mathrm{u,t}^\mathsf{H} +\bar{\Phi}_\mathrm{w}\bar{\Phi}_\mathrm{w}^\mathsf{H}).
\end{align}
\end{lemma}

\ifbool{arxivversion}{The proof of Lemma \ref{lem:phiphitight} may be found in Appendix \ref{app:lem_phiphitight}.}{The proof of Lemma \ref{lem:phiphitight} may be found in \cite[Appendix IV]{venkatasubramanian2025beyond}.} The result in Lemma \ref{lem:phiphitight} depends on the transient error due to the input and the effect of the disturbance, which is unknown. Therefore, in what follows, we derive suitable bounds.

\subsection{Bounds on transient error and the effect of disturbance}\label{sec:te_dist}
Due to bounded inputs and strict stability, the transient error $\bar{x}_\mathrm{u,t}(\omega_i)$ decays uniformly to zero as $T \to \infty$. Let us denote the impulse response of the system with respect to the input $u_k$ as $G(i)=A_\mathrm{tr}^{i-1}B_\mathrm{tr}$. The following lemma provides a deterministic bound on the transient error for finite $T$.

\begin{lemma}\label{lem:transient} \cite[Theorem 2.1]{ljung1999sysid} Let Assumption \ref{a3} hold. We have
\begin{align}\label{eq:norm_phiuerr}
    \sup_{\omega_i \in \Omega_T}\|\bar{x}_{\mathrm{u,t}}(\omega_i)\| \overset{\eqref{eq:min_energy_cost}}{\leq} \frac{2 \gamma_\mathrm{e} G_\mathrm{tr}}{\sqrt{T}}
\end{align}
where $G_\mathrm{tr} = \sum_{i=1}^\infty i \|G(i)\| < \infty$.    
\end{lemma}
Since $A_\mathrm{tr}$ is Schur stable (cf. Assumption \ref{a3}), there exist constants $C \geq 1$ and $\rho \in (0,1)$ such that $\|A_\mathrm{tr}^k\| \leq C \rho^k$, $\forall~k~\in~\mathbb{N}$. From the convergence of arithmetico-geometric series, we have 
\begin{align}\label{eq:Gtr}G_\mathrm{tr}=\sum_{k=1}^\infty k \|A_\mathrm{tr}^k B_\mathrm{tr}\| \leq \|B_\mathrm{tr}\| \sum_{k=1}^\infty k C\rho^k
  = \dfrac{\|B_\mathrm{tr}\| C\rho}{(1-\rho)^2}.
\end{align}

Using \eqref{eq:phibar} and \eqref{eq:barphiw_redef}, we can write $\bar{\Phi}_\mathrm{w}$ as 
\begin{align}\label{eq:Phiwcompact}
    \bar{\Phi}_\mathrm{w}=\begin{bmatrix}   F_{x,\omega_1}A_\mathrm{w}W&\cdots &F_{x,\omega_L}A_\mathrm{w}W\\
    0 & \cdots & 0
    \end{bmatrix}.
\end{align}
Given that $W \sim \mathsf{subG}_{Tn_\mathrm{x}}(\sigma_\mathrm{w}^2I_{Tn_\mathrm{x}})$ with zero mean, \cite[Corollary 1]{ao2025stochastic} implies
\begin{align}\label{eq:Wnorm}
    \mathbb{P} \bigg[ \|W\|^2 \leq \underbrace{\sigma_\mathrm{w}^2 \left((1 + \log 4)T n_\mathrm{x} + 4 \log \delta^{-1}\right)}_{=:\gamma_\mathrm{w}}\bigg] \geq 1-\delta.
\end{align}

\begin{lemma}\label{lem:Phiuerr_Phiw} The matrices $\bar{\Phi}_\mathrm{u,t}$ and $\bar{\Phi}_\mathrm{w}$ satisfy
\begin{align}
\label{eq:phiuerr_bound}
\bar{\Phi}_\mathrm{u,t} \bar{\Phi}_\mathrm{u,t}^\mathsf{H} \preceq\frac{ 4 L \gamma_\mathrm{e}^2}{T} \left(\dfrac{\|B_\mathrm{tr}\| C\rho}{(1-\rho)^2}\right)^2I_{n_\phi} = : \Gamma_\mathrm{u}(\gamma_\mathrm{e}^2),\\
 \label{eq:Gammaw}    \mathbb{P} \left[\bar{\Phi}_\mathrm{w}\bar{\Phi}_\mathrm{w}^\mathsf{H} \preceq \left(\tfrac{L}{T}\|A_\mathrm{w}\|^2 \gamma_\mathrm{w}\right) I_{n_\phi}=:\Gamma_\mathrm{w}\right] \geq 1-\delta.
\end{align}
\end{lemma}
\ifbool{arxivversion}{The proof of Lemma \ref{lem:Phiuerr_Phiw} may be found in Appendix \ref{app:lem_Phiuerr_Phiw}.}{The proof of Lemma \ref{lem:Phiuerr_Phiw} may be found in \cite[Appendix V]{venkatasubramanian2025beyond}.}


\subsection{Convex relaxation}\label{sec:convexrel}
The following proposition provides a sufficient condition linear in $U_\mathrm{e}$.
\begin{proposition}\label{prop:penum_suffcond}
Let Assumptions \ref{a1}, \ref{a3} and \ref{a2} hold. Suppose matrices $U_\mathrm{e}$, $\tilde{U}$ and $\Phi$ satisfy
\begin{align}\label{eq:penum_suffcond}
\nonumber
    (1-\epsilon) V_\mathrm{tr}( U_\mathrm{e} \tilde{U}^\top +  \tilde{U} U_\mathrm{e}^\top - \tilde{U} \tilde{U}^\top) V_\mathrm{tr}^\mathsf{H} & \\
\nonumber
    -  \tfrac{2(1-\epsilon)}{\epsilon} \Gamma_\mathrm{u}(\gamma_\mathrm{e}^2)+ \tfrac{\lambda}{T} I_{n_\phi} - \tfrac{2(1-\epsilon)}{\epsilon} \Gamma_\mathrm{w}& \\ 
    -\tfrac{1}{T} \big( 2C_1 + \JV{2 n_\mathrm{x}} \sigma_\mathrm{w}^2 \log(\det(\bar{D}_T)) + 2 \lambda \bar{\theta}^2 \big) D_\mathrm{des} & \succeq 0.
\end{align}
Then, an estimate $\hat{\theta}_T$ computed as in \eqref{eq:reg_mean_est} upon the application of the input \eqref{eq:exploration_controller} satisfies the exploration goal \eqref{eq:exp_goal2} with probability at least $1-2 \delta$.
\end{proposition}

\ifbool{arxivversion}{The proof of Proposition \ref{prop:penum_suffcond} may be found in Appendix \ref{app:prop_penum_suffcond}.}{The proof of Proposition \ref{prop:penum_suffcond} may be found in \cite[Appendix VI]{venkatasubramanian2025beyond}.} Note that Inequality \eqref{eq:penum_suffcond} in Proposition \ref{prop:penum_suffcond} requires an upper bound on the exploration trajectory data $\bar{D}_T$ (cf. \eqref{eq:reg_mean_est}) in order to derive an LMI for exploration.

\subsection{Upper bound on excitation and $\bar{D}_T$}\label{sec:ubDT}
 
In the following lemma, we provide an upper bound on the term $\log(\det(\bar{D}_T))$.


\begin{lemma}\label{lem:logDT} Let $S_{\textnormal{energy-bound-1}}(\gamma_\mathrm{e},U_\mathrm{e}) \succeq 0$ in \eqref{eq:min_energy_cost} and $\|W\|^2 \leq \gamma_\mathrm{w}$ hold. Then, the matrix $\bar{D}_T$ as in \eqref{eq:reg_mean_est} satisfies
\begin{align}\label{eq:logdetDTub}
\nonumber
    & \log(\det(\bar{D}_T))\\
    \leq &\, n_\theta \log (2T\|A_\mathrm{w}\|^2 \gamma_\mathrm{w} + 2T^2 (\|A_\mathrm{u}\|^2+1) \gamma_\mathrm{e}^2 + \lambda).
\end{align}
\end{lemma}

\ifbool{arxivversion}{The proof of Lemma \ref{lem:logDT} may be found in Appendix \ref{app:lem_logDT}.}{The proof of Lemma \ref{lem:logDT} may be found in \cite[Appendix VII]{venkatasubramanian2025beyond}.} Recall $C_1$ from \eqref{eq:C1}.
Henceforth, we define
\begin{align}\label{eq:C2C3}
\nonumber
    C_2:= &\left( 2C_1 + 2 \lambda \bar{\theta}^2 \right),\\
\nonumber  C_3(\gamma_\mathrm{e}^2):= & \JV{2 n_\mathrm{x}} \sigma_\mathrm{w}^2 n_\theta\log (2T\|A_\mathrm{w}\|^2 \gamma_\mathrm{w})\\
& + \JV{2 n_\mathrm{x}} \sigma_\mathrm{w}^2 n_\theta\log ( 2T^2 (\|A_\mathrm{u}\|^2+1) \gamma_\mathrm{e}^2 + \lambda).
\end{align}

\begin{proposition}\label{prop:preLMI}Let Assumptions \ref{a1}, \ref{a3} and \ref{a2} hold. Suppose matrices $U_\mathrm{e}$, $\tilde{U}$ satisfy
\begin{align}\label{eq:preLMI}
    (1-\epsilon) V_\mathrm{tr}\left( U_\mathrm{e} \tilde{U}^\top +  \tilde{U} U_\mathrm{e}^\top - \tilde{U} \tilde{U}^\top \right) V_\mathrm{tr}^\mathsf{H}+ \tfrac{\lambda}{T} I_{n_\phi}  &\\\nonumber
    - \tfrac{2(1-\epsilon)}{\epsilon} \Gamma_\mathrm{u}(\gamma_\mathrm{e}^2) - \tfrac{2(1-\epsilon)}{\epsilon} \Gamma_\mathrm{w}
    -\tfrac{1}{T}\left(C_2 + C_3 (\gamma_\mathrm{e}^2) \right) D_\mathrm{des}&\succeq 0.
\end{align}
Then, an estimate $\hat{\theta}_T$ computed as in \eqref{eq:reg_mean_est} upon the application of the input \eqref{eq:exploration_controller} satisfies the exploration goal \eqref{eq:exp_goal2} with probability at least $1-2 \delta$.
\end{proposition}
\begin{proof}
    Using $C_2$ and $C_3$ \eqref{eq:C2C3}, and inserting Inequality \eqref{eq:logdetDTub} in \eqref{eq:preLMI} yields condition \eqref{eq:penum_suffcond} in Proposition \ref{prop:penum_suffcond}.
\end{proof}

The result in Proposition \ref{prop:preLMI} depends on the transfer matrix $V_\mathrm{tr}$ which is dependent on the true dynamics $A_\mathrm{tr},\,B_\mathrm{tr}$, and hence uncertain. Therefore, in what follows, we derive suitable bounds.

\subsection{Bounds on transfer matrices and Exploration SDP}\label{sec:expLMI}
Denote 
\begin{align}\label{eq:vtilde}
\tilde{V}&=V_\mathrm{tr}-\hat{V},\\
\label{eq:VhatVxhat}
\hat{V}&=[\hat{V}_{1},\cdots,\hat{V}_{L}] \in \mathbb{C}^{n_\mathrm{\phi}\times Ln_\mathrm{u}},
\end{align}
where $\hat{V}$ is computed using the initial estimates $\hat{\theta}_0=\mathrm{vec}([\hat{A}_0,\hat{B}_0])$ (cf. Assumption \ref{a2}). 
We compute a matrix $\tilde{\Gamma}_\mathrm{V} \succeq 0$ using $\theta_\mathrm{tr} \in \mathbf{\Theta}_0$ such that
\begin{align}\label{eq:tf_prop}
\tilde{V} \tilde{V}^\mathsf{H} \preceq \tilde{\Gamma}_\mathrm{V}.
\end{align}
Since Inequality \eqref{eq:preLMI} in Proposition \ref{prop:preLMI} is nonlinear in the decision variable $\gamma_\mathrm{e}$, we impose an upper bound of the form $\gamma_\mathrm{e}^2 \leq \bar{\gamma}$, where $\bar{\gamma}$ is fixed. Using the Schur complement, this criterion is equivalent to 
\begin{align}\label{eq:energyub}
    S_\textnormal{energy-bound-2}(\gamma_\mathrm{e},\bar{\gamma})=\begin{bmatrix}
        \bar{\gamma} & \gamma_\mathrm{e}\\ \gamma_\mathrm{e} & 1
    \end{bmatrix} \succeq 0.
\end{align}

\begin{table*}
\begin{align}\label{eq:LMI_mid}
\begin{bmatrix}
    V_\mathrm{tr}^\mathsf{H} \\ I
\end{bmatrix}^\mathsf{H}
\begin{bmatrix}
    (1-\epsilon) \left( U_\mathrm{e} \tilde{U}^\top +  \tilde{U} U_\mathrm{e}^\top - \tilde{U} \tilde{U}^\top \right) & 0 \\
    0 & - \tfrac{2(1-\epsilon)}{\epsilon} \Gamma_\mathrm{u}(\bar{\gamma}) - \tfrac{2(1-\epsilon)}{\epsilon} \Gamma_\mathrm{w}+ \tfrac{\lambda}{T} I_{n_\phi} -\tfrac{1}{T}\left( C_2 + C_3(\bar{\gamma}) \right) D_\mathrm{des}
\end{bmatrix}
\begin{bmatrix}
    V_\mathrm{tr}^\mathsf{H} \\ I
\end{bmatrix} \succeq 0
\end{align}
\begin{align}\label{eq:expLMI}
\nonumber
&S_\textnormal{exp}(\epsilon, \lambda, \tau, U_\mathrm{e}, \tilde{U}, \gamma_\mathrm{e}, \bar{\gamma}, \Gamma_\mathrm{u}, \Gamma_\mathrm{w}, \tilde{\Gamma}_\mathrm{v}, \hat{V}, D_\mathrm{des}):=\\
&\begin{bmatrix}
    (1-\epsilon) \left( U_\mathrm{e} \tilde{U}^\top +  \tilde{U} U_\mathrm{e}^\top - \tilde{U} \tilde{U}^\top \right) & 0 \\
    0 &  - \tfrac{2(1-\epsilon)}{\epsilon} \Gamma_\mathrm{u}(\bar{\gamma}) - \tfrac{2(1-\epsilon)}{\epsilon} \Gamma_\mathrm{w}+ \tfrac{\lambda}{T} I_{n_\phi} -\tfrac{1}{T}\left( C_2 + C_3(\bar{\gamma}) \right) D_\mathrm{des}
\end{bmatrix}
- \tau \begin{bmatrix}
        -I & \hat{V}^\mathsf{H}\\ \hat{V} & \tilde{\Gamma}_\mathrm{v} -\hat{V}\hat{V}^\mathsf{H}
    \end{bmatrix} \succeq 0
\end{align}
\end{table*}
The following theorem provides a sufficient condition linear in $U_\mathrm{e}$, which ensures the exploration goal \eqref{eq:exp_goal2}.
\begin{theorem} Let Assumptions \ref{a1}, \ref{a3} and \ref{a2} hold. Suppose the following SDP is feasible:
\begin{align}\label{eq:exp_problem}
    \underset{U_\mathrm{e},\gamma_\mathrm{e}, \atop {\tau \geq 0}}{\inf}  & \; \gamma_\mathrm{e}\\
\nonumber
    \mathrm{s.t. }& \; S_{\textnormal{energy-bound-1}}(\gamma_\mathrm{e},U_\mathrm{e})\succeq 0\\
\nonumber
    & \; S_\textnormal{energy-bound-2}(\gamma_\mathrm{e},\bar{\gamma}) \succeq 0\\
    & \; S_\textnormal{exp}(\epsilon, \lambda, \tau, U_\mathrm{e}, \tilde{U}, \gamma_\mathrm{e}, \bar{\gamma}, \Gamma_\mathrm{u}, \Gamma_\mathrm{w}, \tilde{\Gamma}_\mathrm{v}, \hat{V}, D_\mathrm{des}) \succeq 0\nonumber
\end{align}
where $S_{\textnormal{energy-bound-1}}$, $S_\textnormal{energy-bound-2}$ and $S_\textnormal{exp}$ are defined in \eqref{eq:min_energy_cost}, \eqref{eq:energyub} and \eqref{eq:expLMI}, respectively.
Then, an estimate $\hat{\theta}_T$ computed as in \eqref{eq:reg_mean_est} upon the application of the input \eqref{eq:exploration_controller} satisfies the exploration goal \eqref{eq:exp_goal2} with probability at least $1-2\delta$.
\end{theorem}
\begin{proof}
By using \eqref{eq:vtilde}, Inequality \eqref{eq:tf_prop} can be written as
\begin{align}\label{eq:Vtr_ineq}
    \begin{bmatrix}
        V_\mathrm{tr}^\mathsf{H} \\ I   \end{bmatrix}^\mathsf{H}
        \begin{bmatrix}
        -I & \hat{V}^\mathsf{H}\\ \hat{V} & \tilde{\Gamma}_\mathrm{v} -\hat{V}\hat{V}^\mathsf{H}
    \end{bmatrix}
    \begin{bmatrix}
        V_\mathrm{tr}^\mathsf{H} \\ I
    \end{bmatrix} \succeq 0.
\end{align}
By using the matrix S-lemma \cite{vanwaarde2022noisy}, Inequality \eqref{eq:LMI_mid} holds for all $V_\mathrm{tr}$ satisfying \eqref{eq:Vtr_ineq} if $S_\textnormal{exploration}(\epsilon, \lambda, \tau, U_\mathrm{e}, \tilde{U}, \gamma_\mathrm{e}, \bar{\gamma}, \Gamma_\mathrm{u}, \Gamma_\mathrm{w}, \tilde{\Gamma}_\mathrm{v}, \hat{V}, D_\mathrm{des}) \succeq 0$ \eqref{eq:expLMI} with $\tau \geq 0$. Inequality \eqref{eq:LMI_mid} can be written as
\begin{align}
    (1-\epsilon) V_\mathrm{tr}\left( U_\mathrm{e} \tilde{U}^\top +  \tilde{U} U_\mathrm{e}^\top - \tilde{U} \tilde{U}^\top \right) V_\mathrm{tr}^\mathsf{H}+ \tfrac{\lambda}{T} I_{n_\phi}  &\\
    - \tfrac{2(1-\epsilon)}{\epsilon} \Gamma_\mathrm{u}(\bar{\gamma}) - \tfrac{2(1-\epsilon)}{\epsilon} \Gamma_\mathrm{w}
    -\tfrac{1}{T}\left(C_2 + C_3 (\bar{\gamma}) \right) D_\mathrm{des}&\succeq 0.\nonumber
\end{align}
Inserting the inequality $\gamma_\mathrm{e}^2 \leq \bar{\gamma}$, i.e., $S_\textnormal{energy-bound-2}(\gamma_\mathrm{e},\bar{\gamma}) \succeq 0$, yields condition \eqref{eq:preLMI} in Proposition \ref{prop:preLMI}. Hence, an estimate $\hat{\theta}_T$ computed as in \eqref{eq:reg_mean_est} upon the application of the solution of Problem \eqref{eq:exp_problem}, $U_\mathrm{e}$, satisfies the exploration goal \eqref{eq:exp_goal2}
with probability at least $1-2\delta$.
\end{proof}

A solution of \eqref{eq:exp_problem} yields $U_\mathrm{e}=\mathrm{diag}(\bar{u}(\omega_1),\dots,\bar{u}(\omega_L))$, i.e., the exploration input, which guarantees the desired uncertainty bound $D_\mathrm{des}$ \eqref{eq:exp_goal2}. In LMI \eqref{eq:expLMI}, the terms $\Gamma_\mathrm{u}$, $\Gamma_w$, $C_3$ and $G_\mathrm{tr}$ comprise constants $\|A_\mathrm{w}\|$, $\|A_\mathrm{u}\|$, $\|B_\mathrm{tr}\|$, $C$ and $\rho$ which are uncertain since they depend on the true dynamics $A_\mathrm{tr}, B_\mathrm{tr}$. Given $\theta_\mathrm{tr} \in \Theta_0$, such constants can be computed using robust control tools or scenario optimization as in \ifbool{arxivversion}{Appendix \ref{app:scenario}}{\cite[Appendix VIII]{venkatasubramanian2025beyond}}. Furthermore, the matrix $\tilde{\Gamma}_\mathrm{V} \succ 0$ may be computed using robust control methods, cf.  \cite[Appendices A, C]{venkatasubramanian2023sequential}, since \eqref{eq:tf_prop} is an LMI.
Note that imposing the upper bound $\bar{\gamma} \geq \gamma_\mathrm{e}^2$ and the convex relaxation procedure introduces suboptimality in the solution of \eqref{eq:exp_problem} This suboptimality can be reduced by iterating \eqref{eq:exp_problem} multiple times by re-computing $\tilde{U}$ and $\bar{\gamma}$, until they do not change, for the next iteration as
\begin{align}\label{eq:L_choice}
\tilde{U}=U_\mathrm{e}^*,\;\;\;
\bar{\gamma}= (\gamma_\mathrm{e}^*)^2,
\end{align}
where $U_\mathrm{e}^*, \gamma_\mathrm{e}^*$ is the solution from the previous iteration. Since the previous optimal solution $U_\mathrm{e}^*$ remains feasible in the next iteration, $\gamma_\mathrm{e}$ is guaranteed to be non-increasing. From LMI \eqref{eq:expLMI}, it can be inferred that $\gamma_\mathrm{e}^2$ scales as $\frac{1}{T}$. As $D_\mathrm{des}$ and $T$ grow proportionally, $U_\mathrm{e}$ remains nearly unchanged. The proposed targeted exploration strategy is outlined in Algorithm \ref{alg:main}. Note that the proposed algorithm directly ensures the desired accuracy of the parameters and requires no iterative experiments.

\begin{algorithm}[h]
\caption{Targeted exploration}
\label{alg:main}
\begin{algorithmic}[1]
\State Specify exploration length $T$, frequencies $\omega_i,\,i=1,...,L$, initial estimates $\hat{A}_0,\,\hat{B}_0$ and uncertainty level $D_0$, desired accuracy of parameters $D_\mathrm{des}$, set $\epsilon, \lambda$, probability level $\delta$. 

\State Compute $\hat{V}$ \eqref{eq:VhatVxhat} using the initial estimates and $\tilde{\Gamma}_\mathrm{v}$ \eqref{eq:tf_prop}. Compute constants $C_1$ \eqref{eq:C1}, $\|A_\mathrm{w}\|$, $\|A_\mathrm{u}\|$, $\|B_\mathrm{tr}\|$, $C$ and $\rho$ for $C_2$ and $C_3$ \eqref{eq:C2C3},  $\Gamma_\mathrm{u}$ \eqref{eq:phiuerr_bound}, $\Gamma_\mathrm{w}$ \eqref{eq:Gammaw} and $G_\mathrm{tr}$ \eqref{eq:Gtr}

\State Select initial candidates $\tilde{U}$ and $\bar{\gamma}$ \eqref{eq:L_choice}.

\State Set tolerance $\mathrm{tol}>0$.

\While{$\lvert \frac{\gamma_\mathrm{e} - \gamma_\mathrm{e}^*}{\gamma_\mathrm{e}}\rvert \geq \mathrm{tol}$}
\State Solve the optimization problem \eqref{eq:exp_problem}.

\State Update $\tilde{U}$ and $\bar{\gamma}$ \eqref{eq:L_choice}.
\EndWhile

\State Apply the exploration input \eqref{eq:exploration_controller} for $k=0,...,T-1$.

\State Compute parameter estimate $\hat{\theta}_T$ \eqref{eq:reg_mean_est}.

\end{algorithmic}
\end{algorithm}

\subsection{Discussion}
The proposed approach is related to the methods in \cite{umenberger2019robust, ferizbegovic2019learning, barenthin2008identification}, which assume i.i.d.~Gaussian disturbances with zero mean. In contrast, our approach considers sub-Gaussian disturbances, which is a broader class that includes i.i.d.~Gaussian as a special case, as well as other light-tailed distributions like uniform distributions. In contrast to exploration methods \cite{umenberger2019robust, ferizbegovic2019learning, barenthin2008identification} with asymptotic guarantees, we derive an exploration strategy with non-asymptotic guarantees. Unlike the classical stochastic bound \cite{umenberger2019robust}, the term $(R(\tilde{D}_T)^\frac{1}{2} + \lambda^\frac{1}{2} \bar{\theta})^2$ appearing in Theorem \ref{thm:RLS_suffcond} that scales the uncertainty bound is also data-dependent. This term presents significant additional challenges in input design, as it introduces yet another coupling between the input design and the resulting uncertainty bound. The derived exploration strategy is also related to recent works \cite{wagenmaker2020active, sarker2020parameter}, which consider periodic/sinusoidal inputs and frequency-domain analysis for parameter estimation. However, these methods neglect transient effects~\cite{ljung1999sysid} during input design and analysis, making their guarantees reliant on steady-state responses, which may lead to impractically long experiments in many applications. On the contrary, we explicitly account for the transient error due to the input, the effect of the disturbances, and the initial parametric uncertainty to yield a robust exploration strategy. 



\section{Numerical Example}\label{sec:numerical}

In this section, we demonstrate the applicability of the exploration strategy through a numerical example. Numerical simulations\footnote{\footnotesize{The source code for the simulations is available at https:$//$github.com$/$jananivenkatasubramanian$/$NonAsymptoticTE}} were performed on MATLAB using CVX \cite{cvx} in conjunction with the default SDP solver SDPT3. We consider a linear system \eqref{eq:sys} with
\begin{align}\label{eq:exsys}
A_\mathrm{tr}=\begin{bmatrix}
0.49 & 0.49 & 0 \\ 0 & 0.49 & 0.49 \\ 0 & 0 & 0.49  
\end{bmatrix},\;
B_\mathrm{tr}=\begin{bmatrix}
0 \\ 0 \\ 0.49
\end{bmatrix}
\end{align}
which belongs to a class of systems identified as `hard to learn' \cite{tsiamis2021linear}. We select $L=5$ frequencies $\omega_i\in\{0,0.1,0.2,0.3,0.4\}$ and the initial estimate as $\hat{\theta}_0=\theta_\mathrm{tr}+(4 \times 10^{-4})\mathbf{1}_{n_\mathrm{x}n_\phi}$ with initial uncertainty level $D_0 = 10^3 I_{n_\phi}$. We set $\epsilon=0.5$, $\lambda =1$, $\delta = 0.05$ and $\sigma_\mathrm{w}=0.01$. In what follows, we study the effect of the finite duration $T$ of the experiment on the average input energy required.
\begin{figure}[t!]
\begin{center}
\includegraphics[width=0.5\textwidth]{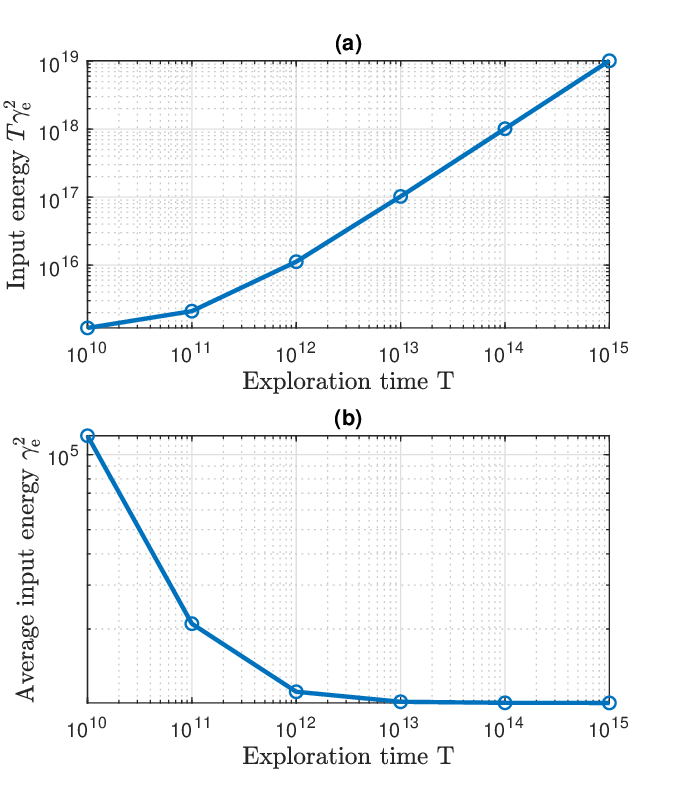}
\end{center}
\caption{Illustration of (a) the exploration input energy $T\gamma_\mathrm{e}^2$, (b) the average input energy $\gamma_\mathrm{e}^2$, in comparison with the exploration time $\gamma_\mathrm{w}$ for the initial uncertainty level $\|D_0\|=10^3$ and $\frac{\|D_\mathrm{des}\|}{T}=10^{-10}$, i.e., the desired accuracy is increased linearly with the exploration time $T$.}
\label{fig:ge-T} 
\end{figure}
\textit{Non-asymptotic behaviour: } In the asymptotic case where $T \to \infty$ \cite{venkatasubramanian2023sequential}, the required input energy $T\gamma_\mathrm{e}^2$ scales linearly with the desired accuracy $\|D_\mathrm{des}\|$. In this example, we study the effect of the exploration time $T$ and the transient error decay on the required input energy $\gamma_\mathrm{e}^2$. For this purpose, we run six trials for the following exploration times $T \in \{10^{10},10^{11},10^{12},10^{13},10^{14},10^{15}\}$. For each trial, we scale $D_\mathrm{des}$ proportionately, yielding $D_\mathrm{des} \in \{10^0,10^1,10^2,10^3,10^4,10^5 \}$, respectively. Each trial comprises executing Algorithm \ref{alg:main} to obtain the exploration inputs \eqref{eq:exploration_controller} and average input energy $\gamma_\mathrm{e}^2$. In Figure \ref{fig:ge-T}, we see that for $T >> 10^{12}$, we approximately recover this linear relationship, while for $T < 10^{11}$, we see that the transient effects are significant and more input energy is required. This particularly highlights the role of transient effects in finite-time experiment design, indicating the need for careful input design in order to mitigate transient effects even for large times $T$.

\section{Conclusion}
We presented a non-asymptotic targeted exploration strategy for linear systems subject to general sub-Gaussian disturbances. As the main result, we derived LMIs that guarantee an $\textit{a priori}$ error-bound on the estimated parameters after finite-time exploration with high probability. We explicitly accounted for transient effects and parametric uncertainty in the analysis and input design. Through a numerical example, we demonstrated the role of transient effects in non-asymptotic experiment design.

\bibliographystyle{ieeetr}
\bibliography{lit} 

\ifbool{arxivversion}{
\appendices


\section{Proof of Theorem \ref{thm:RLS_suffcond}}\label{app:thm_RLS_suffcond}
\begin{proof}
By applying the Schur complement twice to the condition in \eqref{eq:nonreg_ellipsoid}, we have
\begin{align}\label{eq:suffcondrls_p1}
\nonumber
    &(\hat{\theta}_T-\theta_\mathrm{tr}) (\hat{\theta}_T-\theta_\mathrm{tr})^\top\\
\nonumber
    \preceq& \left(R(\bar{D}_T)^\frac{1}{2} + \lambda^\frac{1}{2} \bar{\theta} \right)^2 \bar{D}_T^{-1}\\
    \overset{\eqref{eq:reg_mean_est}}{=}& \left(R(\bar{D}_T)^\frac{1}{2} + \lambda^\frac{1}{2} \bar{\theta} \right)^2 ((\Phi \Phi^\top) \otimes I_{n_\mathrm{x}} + \lambda I_{n_\phi n_\mathrm{x}})^{-1}
\end{align}
with probability at least $1-\delta$. Inequality \eqref{eq:suffcond_rls} can be written as
\begin{align}\label{eq:suffcondrls_p2}
\nonumber
    & \left((\Phi \Phi^\top) \otimes I_{n_\mathrm{x}} + \lambda I_{n_\phi n_\mathrm{x}}\right)^{-1}\\
    \preceq & \left(R(\bar{D}_T)^\frac{1}{2} + \lambda^\frac{1}{2} \bar{\theta} \right)^{-2}  (D_\mathrm{des} \otimes I_{n_\mathrm{x}})^{-1}.
\end{align}
By inserting \eqref{eq:suffcondrls_p2} in \eqref{eq:suffcondrls_p1}, we get
\begin{align}\label{eq:suffcondrls_p3}
    (\theta-\hat{\theta}_T) (\theta-\hat{\theta}_T)^\top \preceq (D_\mathrm{des}\otimes I_{n_\mathrm{x}})^{-1}
\end{align}
with probability at least $1-\delta$. Finally, applying the Schur complement twice to \eqref{eq:suffcondrls_p3} yields the exploration goal \eqref{eq:exp_goal2}.
\end{proof}

\section{Proof of Lemma \ref{lem:boundRDT}}\label{app:lem_boundRDT}
\begin{proof}\JV{We first derive an upper bound on the term $R(\bar{D}_T)$ as follows:
\begin{align}\label{eq:boundRDT}
\nonumber
    & R(\bar{D}_T)\overset{\eqref{eq:nonreg-vector}}{=} n_\mathrm{x} \sigma_\mathrm{w}^2 \log \left(
\frac{n_\mathrm{x}^2 \det (\bar{D}_T)}{\delta^2 \det(\lambda I_{n_\mathrm{x} n_\phi})} \right)\\
\nonumber
    = &  n_\mathrm{x} \sigma_\mathrm{w}^2 \left( \log \left( \frac{n_\mathrm{x}^2}{\delta^{2}}\right) - \log (\det(\lambda I_{n_\theta}))  + \log(\det(\bar{D}_T)) \right) \\
    = & \underbrace{n_\mathrm{x} \sigma_\mathrm{w}^2 \left( \log \left( \frac{n_\mathrm{x}^2}{\delta^{2}}\right) -n_\theta\log (\lambda) \right)}_{C_1} + n_\mathrm{x} \sigma_\mathrm{w}^2 \log(\det(\bar{D}_T)).
\end{align}
Finally, by Young's inequality \cite{caverly2019lmi}, we have that
\begin{align*}
\nonumber
(R(\bar{D}_T)^\frac{1}{2} + \lambda^\frac{1}{2} \bar{\theta})^2 \leq  & ( 2 R(\bar{D}_T) + 2 \lambda \bar{\theta}^2 ) \\
\overset{\eqref{eq:boundRDT}}{\leq}  &  ( 2C_1 +2 n_\mathrm{x}\sigma_\mathrm{w}^2 \log(\det(\bar{D}_T)) + 2 \lambda \bar{\theta}^2 )
\end{align*}    
which yields Inequality \eqref{eq:bound_RDT_lambda}.}
\end{proof}

\section{Proof of Proposition \ref{prop:RLS_suffcond_RDTbound}}\label{app:prop_RLS_sc}
\begin{proof}
Starting from \eqref{eq:suffcond_prop}, we have
\begin{align}
\nonumber
    0 \preceq & \Phi \Phi^\top + \lambda I_{n_\phi}\\
\nonumber
    & -\left( 2C_1 + \JV{2 n_\mathrm{x}} \sigma_\mathrm{w}^2 \log(\det(\bar{D}_T)) + 2 \lambda \bar{\theta}^2 \right) D_\mathrm{des}\\
\nonumber \overset{\eqref{eq:bound_RDT_lambda}}{\preceq}& \Phi \Phi^\top + \lambda I_{n_\phi} - \left(R(\bar{D}_T)^\frac{1}{2} + \lambda^\frac{1}{2} \bar{\theta} \right)^2 D_\mathrm{des},
\end{align}
which yields condition \eqref{eq:suffcond_rls} in Theorem \ref{thm:RLS_suffcond}. Hence, the exploration goal \eqref{eq:exp_goal2} is achieved with probability at least $1-\delta$.
\end{proof}

\section{Proof of Lemma \ref{lem:phiphitight}}\label{app:lem_phiphitight}
\begin{proof}
From the Parseval-Plancherel identity, we have
\begin{align*}
\nonumber
    \frac{1}{T}\Phi \Phi^\top &\overset{\eqref{eq:spectral_amplitude}}{=} \sum_{i=1}^{T} \bar{\phi}(\omega_i) \bar{\phi}(\omega_i)^\mathsf{H}\succeq  \sum_{i=1}^{L} \bar{\phi}(\omega_i) \bar{\phi}(\omega_i)^\mathsf{H}\\
\nonumber
    & \overset{\eqref{eq:phibarcompact}}{=}(V_\mathrm{tr}U_\mathrm{e}+\bar{\Phi}_\mathrm{u,t}+\bar{\Phi}_\mathrm{w}) (V_\mathrm{tr}U_\mathrm{e}+\bar{\Phi}_\mathrm{u,t}+\bar{\Phi}_\mathrm{w})^\mathsf{H}.
\end{align*}
For any $\epsilon >0$, applying Young's inequality \cite{caverly2019lmi} twice gives
\begin{align*}
\nonumber
& (V_\mathrm{tr}U_\mathrm{e}+\bar{\Phi}_\mathrm{u,t}+\bar{\Phi}_\mathrm{w}) (V_\mathrm{tr}U_\mathrm{e}+\bar{\Phi}_\mathrm{u,t}+\bar{\Phi}_\mathrm{w})^\mathsf{H} \\
\nonumber
 \succeq & 
 (1-\epsilon)V_\mathrm{tr}U_\mathrm{e}U_\mathrm{e}^\top V_\mathrm{tr}^\mathsf{H}-\left(\tfrac{1-\epsilon}{\epsilon}\right)(\bar{\Phi}_\mathrm{u,t}+\bar{\Phi}_\mathrm{w}) (\bar{\Phi}_\mathrm{u,t}+\bar{\Phi}_\mathrm{w})^\mathsf{H} \\  
 \succeq & 
 (1-\epsilon)V_\mathrm{tr}U_\mathrm{e}U_\mathrm{e}^\top V_\mathrm{tr}^\mathsf{H}-2\left(\tfrac{1-\epsilon}{\epsilon}\right)(\bar{\Phi}_\mathrm{u,t}\bar{\Phi}_\mathrm{u,t}^\mathsf{H} +\bar{\Phi}_\mathrm{w}\bar{\Phi}_\mathrm{w}^\mathsf{H}),
 \end{align*}
which yields Inequality \eqref{eq:phiphitight}.
\end{proof}

\section{Proof of Lemma \ref{lem:Phiuerr_Phiw}}\label{app:lem_Phiuerr_Phiw}
\begin{proof} Using \eqref{eq:norm_phiuerr} from Lemma \ref{lem:transient} and \eqref{eq:Gtr}, we have
\begin{align*}
\nonumber
\bar{\Phi}_\mathrm{u,t}\bar{\Phi}_\mathrm{u,t}^\mathsf{H} & \preceq \left(\sum_{i=1}^L \max_{\omega_i \in \Omega_T}\|\bar{x}_{\mathrm{u,t}}(\omega_i)\|^2 \right)I_{n_\phi}\overset{\eqref{eq:norm_phiuerr}}{\preceq}  \frac{ 4 L \gamma_\mathrm{e}^2 G_\mathrm{tr}^2}{T} I_{n_\phi}\\
& \overset{\eqref{eq:Gtr}}{\preceq} \frac{ 4 L \gamma_\mathrm{e}^2}{T} \left(\dfrac{\|B_\mathrm{tr}\| C\rho}{(1-\rho)^2}\right)^2I_{n_\phi},
\end{align*}
which yields \eqref{eq:phiuerr_bound}. By using the sub-multiplicativity property of the operator norm, we have
\begin{align}
\nonumber    \bar{\Phi}_\mathrm{w}\bar{\Phi}_\mathrm{w}^\mathsf{H} \overset{\eqref{eq:Phiwcompact}}{\preceq} & \left(\max_{\omega_i \in \Omega_T} \sqrt{L} \|F_{\omega_i} A_\mathrm{w} W\|\right)^2 I_{n_\phi} \\
\nonumber
    \overset{\eqref{eq:supFwi}}{\preceq} & \left( \tfrac{L}{T}  \|A_\mathrm{w}\|^2 \|W\|^2\right) I_{n_\phi}    \overset{ \eqref{eq:Wnorm}}{\preceq} \left(\tfrac{L}{T} \|A_\mathrm{w}\|^2 \gamma_\mathrm{w}\right) I_{n_\phi},
\end{align}
with probability at least $1-\delta$, which yields \eqref{eq:Gammaw}.
\end{proof}

\section{Proof of Proposition \ref{prop:penum_suffcond}}\label{app:prop_penum_suffcond}
\begin{proof}
The proof is divided into two parts. In the first part, we show that condition \eqref{eq:penum_suffcond} implies condition \eqref{eq:suffcond_prop} in Proposition \ref{prop:RLS_suffcond_RDTbound}. In the second part, we derive joint probabilistic bounds for the Inequality \eqref{eq:penum_suffcond}.

\textbf{Part I. } From \cite[Lemma 12]{venkatasubramanian2024robust}, we have the following convex relaxation: 
\begin{align}\label{eq:conv_rel}
V_\mathrm{tr} U_\mathrm{e} U_\mathrm{e}^\top {V_\mathrm{tr}}^\mathsf{H}  \succeq V_\mathrm{tr}\left( U_\mathrm{e} \tilde{U}^\top +  \tilde{U} U_\mathrm{e}^\top - \tilde{U} \tilde{U}^\top \right)V_\mathrm{tr}^\mathsf{H}.
\end{align}
Starting from \eqref{eq:penum_suffcond}, we have 
\begin{align}\label{eq:penum_mid}
    & 0\\
\nonumber
    \preceq &(1-\epsilon) V_\mathrm{tr}\left( U_\mathrm{e} \tilde{U}^\top +  \tilde{U} U_\mathrm{e}^\top - \tilde{U} \tilde{U}^\top \right) V_\mathrm{tr}^\mathsf{H}\\
\nonumber
    & - \frac{2(1-\epsilon)}{\epsilon} \Gamma_\mathrm{u}(\gamma_\mathrm{e}^2) - \frac{2(1-\epsilon)}{\epsilon} \Gamma_\mathrm{w}+ \frac{\lambda}{T} I_{n_\phi}\\
\nonumber
    & -\frac{1}{T} \big( 2C_1 + \JV{2 n_\mathrm{x}} \sigma_\mathrm{w}^2 \log(\det(\bar{D}_T)) + 2 \lambda \bar{\theta}^2 \big) D_\mathrm{des}\\
\nonumber
    \overset{\eqref{eq:conv_rel}}{\preceq}& (1-\epsilon) V_\mathrm{tr} U_\mathrm{e} U_\mathrm{e}^\top V_\mathrm{tr}^\mathsf{H} - \frac{2(1-\epsilon)}{\epsilon} \Gamma_\mathrm{u}(\gamma_\mathrm{e}^2)- \frac{2(1-\epsilon)}{\epsilon} \Gamma_\mathrm{w}\\
\nonumber
    & + \frac{\lambda}{T} I_{n_\phi}
    -\frac{1}{T}\left( 2C_1 + \JV{2 n_\mathrm{x}}\sigma_\mathrm{w}^2 \log(\det(\bar{D}_T)) + 2 \lambda \bar{\theta}^2 \right) D_\mathrm{des}\\
\nonumber
    \overset{\eqref{eq:phiuerr_bound}, \atop \eqref{eq:Gammaw}}{\preceq} & (1-\epsilon) V_\mathrm{tr} U_\mathrm{e} U_\mathrm{e}^\top V_\mathrm{tr}^\mathsf{H}  - \frac{2(1-\epsilon)}{\epsilon} \bar{\Phi}_\mathrm{u,t}\bar{\Phi}_\mathrm{u,t}^\mathsf{H}\\
\nonumber
    & - \frac{2(1-\epsilon)}{\epsilon} \bar{\Phi}_\mathrm{w}\bar{\Phi}_\mathrm{w}^\mathsf{H}  + \frac{\lambda}{T} I_{n_\phi} \\
\nonumber
    &-\frac{1}{T}\left( 2C_1 + \JV{2 n_\mathrm{x}} \sigma_\mathrm{w}^2 \log(\det(\bar{D}_T)) + 2 \lambda \bar{\theta}^2 \right) D_\mathrm{des}\\
\nonumber
    \overset{\eqref{eq:phiphitight}}{\preceq}& \frac{1}{T} (\Phi \Phi^\top + \lambda I_{n_\phi})\\
\nonumber
    &-\frac{1}{T}(\left( 2C_1 + \JV{2 n_\mathrm{x}} \sigma_\mathrm{w}^2 \log(\det(\bar{D}_T)) + 2 \lambda \bar{\theta}^2 \right) D_\mathrm{des}).
\end{align}
Multiplying \eqref{eq:penum_mid} by $T$ yields \eqref{eq:suffcond_prop} in Proposition \ref{prop:RLS_suffcond_RDTbound}.

\textbf{Part II.} 
We have
\begin{align}\label{eq:jointprob}
\nonumber
    & \mathbb{P}\left[ \|\hat{\theta}_T - \theta_\mathrm{tr}\|_{(D_\mathrm{des}\otimes I_{n_\mathrm{x}})} \leq 1 \right]\\
\nonumber
    \geq &\mathbb{P}\left[(\hat{\theta}_\mathrm{tr} \in \Theta_T)\cap (\|W\|^2\leq \gamma_\mathrm{w})\right]\\ \geq
    & 1 - \mathbb{P}\left[\hat{\theta}_\mathrm{tr} \notin \Theta_T\right] - \mathbb{P}\left[\|W\|^2\nleq \gamma_\mathrm{w}\right]  \overset{\eqref{eq:nonreg_ellipsoid}, \atop \eqref{eq:Wnorm}}{\geq}  1-2\delta,
\end{align}
wherein the penultimate inequality follows from De Morgan's law.

In Part I, assuming $\hat{\theta}_\mathrm{tr} \in \Theta_T$ \eqref{eq:nonreg_ellipsoid} and $\|W\|^2 \leq \gamma_\mathrm{w}$ \eqref{eq:Wnorm}, Inequality \eqref{eq:penum_suffcond} implies \eqref{eq:exp_goal2}. Since both $\hat{\theta}_\mathrm{tr} \in \Theta_T$ and $\|W\|^2 \leq \gamma_\mathrm{w}$ hold jointly with probability at least $1-2\delta$ as shown in Part II, Inequality \eqref{eq:penum_suffcond} implies \eqref{eq:exp_goal2} with probability at least $1-2\delta$.
\end{proof}

\section{Proof of Lemma \ref{lem:logDT}
}\label{app:lem_logDT}
\begin{proof}
With $L=T$, and using \eqref{eq:barphiu_redef} and \eqref{eq:barphiw_redef}, we have
\begin{align}
\nonumber
 \bar{\Phi}_\mathrm{u}&= \begin{bmatrix}F_{x,\omega_1}A_\mathrm{u}U&\cdots&F_{x,\omega_{T}}A_\mathrm{u}U \\ F_{u,\omega_1}U&\cdots&F_{u,\omega_{T}}U
 \end{bmatrix},\\
 \nonumber \bar{\Phi}_\mathrm{w}&=\begin{bmatrix}F_{x,\omega_1}A_\mathrm{w}W & \cdots & F_{x,\omega_{T}}A_\mathrm{w}W \\ 0 & \cdots & 0\end{bmatrix}.
\end{align}
Similar to \eqref{eq:Gammaw}, using the sub-multiplicativity property of the operator norm, \eqref{eq:upars}, and \eqref{eq:supFwi}, we have
\begin{align}\label{eq:phiunorm}
\nonumber
    \bar{\Phi}_\mathrm{u} \bar{\Phi}_\mathrm{u}^\mathsf{H} &\preceq \max_{\omega_i \in \Omega_T} T\left( \|F_{x,\omega_i} A_\mathrm{u} U\|^2 + \|F_{u,\omega_i} U\|^2 \right) I_{n_\phi}\\
    &\preceq T \left(\|A_\mathrm{u}\|^2+1 \right)\gamma_\mathrm{e}^2  I_{n_\phi}.
\end{align}
By the Parseval-Plancheral identity, Young's inequality \cite{caverly2019lmi}, Inequalities \eqref{eq:Gammaw}, \eqref{eq:phiunorm}, and with $L=T$, we have
\begin{align*}
\nonumber
\Phi \Phi^\top &= T \sum_{i=1}^{T}\bar{\phi}(\omega_i) \bar{\phi}(\omega_i)^\mathsf{H}\\
\nonumber
&= T \left(\bar{\Phi}_\mathrm{u}+\bar{\Phi}_\mathrm{w}\right) \left(\bar{\Phi}_\mathrm{u}+\bar{\Phi}_\mathrm{w}\right)^\mathsf{H}\\
\nonumber
 & \preceq  2T\left(\bar{\Phi}_\mathrm{u}\bar{\Phi}_\mathrm{u}^\mathsf{H} + \bar{\Phi}_\mathrm{w}\bar{\Phi}_\mathrm{w}^\mathsf{H}\right)\\
 & \overset{\eqref{eq:Gammaw}, \atop \eqref{eq:phiunorm}}{\preceq} 2T(\|A_\mathrm{w}\|^2 \gamma_\mathrm{w} + T (\|A_\mathrm{u}\|^2+1) \gamma_\mathrm{e}^2)I_{n_\phi}.
\end{align*}
Hence, we have
\begin{align*}
\nonumber
    \bar{D}_T&\overset{\eqref{eq:reg_mean_est}}{=}(\Phi \Phi^\top)\otimes I_{n_\mathrm{x}} + \lambda I_{n_\theta}\\
    &\preceq \left( 2T\|A_\mathrm{w}\|^2 \gamma_\mathrm{w} + 2T^2 (\|A_\mathrm{u}\|^2+1) \gamma_\mathrm{e}^2 + \lambda \right) I_{n_\theta},
\end{align*}
which yields \eqref{eq:logdetDTub}.
\end{proof}

\section{Sample-based constants}\label{app:scenario}
In order to estimate $G_\mathrm{tr}$, $\|A_\mathrm{w}\|^2$ and $\|A_\mathrm{u}\|^2$, we generate $N_\mathrm{s}$ samples of $\mathrm{vec}([A_i,B_i])=\theta_i \in \mathbf{\Theta}_0$, $i=1,\dots, N_\mathrm{s}$, (cf. Assumption \ref{a2}). Given a probability of violation $\delta$, confidence $1-\beta$, and the number of uncertain decision variables $d=1$, a lower bound on the number of samples $N_\mathrm{s}$ required to estimate $G_\mathrm{tr}$, $\|A_\mathrm{w}\|^2$ or $\|A_\mathrm{u}\|^2$ with confidence $1-\beta$ is given as \cite{campi2009scenario}: $N_\mathrm{s} \geq \frac{2}{\delta}\left( \ln \frac{1}{\beta} + d \right)$.
\subsubsection*{Estimate of $\gamma_\mathrm{G} \geq G_\mathrm{tr}$} 
Since $A_\mathrm{tr}$ is Schur stable (cf. Assumption \ref{a3}), we can determine constants $C$ and $\rho$ such that $\|A_i^k\| \leq C \rho^k,\, \forall i=1,...,N_\mathrm{s}$. Denote $\bar{B}=\max_{i=1,...,N_\mathrm{s}}\|B_i\|$. From the convergence of arithmetico-geometric series, we have 
\begin{align}
    G_\mathrm{tr}=\sum_{k=1}^\infty k \|A_\mathrm{tr}^k B_\mathrm{tr}\| \leq \bar{B}\sum_{k=1}^\infty k C\rho^k
  = \dfrac{\bar{B} C\rho}{(1-\rho)^2}=:\gamma_\mathrm{G}.
\end{align}



\subsubsection*{Estimates of $\gamma_\mathrm{A_w} \geq \|A_\mathrm{w}\|^2$ and $\gamma_\mathrm{A_u} \geq \|A_\mathrm{u}\|^2$} We construct $A_{\mathrm{w},i}$ and $A_{\mathrm{u},i}$ using $A_i$, $i=1,...,N_\mathrm{s}$, according to \eqref{eq:barphiw_redef} and \eqref{eq:barphiu_redef}, respectively. We have
\begin{align}
\gamma_\mathrm{A_w}=\max_{i=1,...,N_\mathrm{s}}\|A_{\mathrm{w},i}\|^2,\quad \gamma_\mathrm{A_u}&=\max_{i=1,...,N_\mathrm{s}}\|A_{\mathrm{u},i}\|^2.
\end{align}

}{}

\end{document}